\documentclass[
12pt,
prl,
reprint,
a4paper,
longbibliography,
nofootinbib,
nobibnotes,
superscriptaddress
]{revtex4-2}
\usepackage[utf8]{inputenc}
\usepackage[T1]{fontenc}
\usepackage{hyperref}
\usepackage{amsmath,amssymb,amsthm,amsfonts,mathtools}
\usepackage{multirow}
\usepackage{array}
\usepackage[dvipsnames]{xcolor}

\hypersetup{
    colorlinks = true,
    citecolor=PineGreen,
    linkcolor=MidnightBlue,
    urlcolor=PineGreen
}

\theoremstyle{definition}
\newtheorem{theorem}{Theorem}

\newtheorem{corollary}[theorem]{Corollary}

\newtheorem{lemma}[theorem]{Lemma}

\DeclareMathOperator\tr{Tr}

\def\id{\mathbb{I}}%
\def\var{\mathrm{Var}}
\def\cov{\mathrm{Cov}}

\def\stab{\operatorname{STAB}}
\def\th{\operatorname{TH}}
\def\be{\operatorname{BETA}}

\begin{document}

\author{Zhen-Peng Xu}%
\affiliation{School of Physics, Anhui University, Hefei 230601, China}

\author{Qi Ye}%
\affiliation{Institute for Interdisciplinary Information Sciences, Tsinghua University, Beijing 100084, China}

\author{Jie Wang}%
\affiliation{State Key Laboratory of Mathematical Sciences, Academy of Mathematics and Systems Science, Chinese Academy of Sciences, Beijing, China}

\author{Gereon Ko{\ss}mann}%
\affiliation{Institute for Quantum Information, RWTH Aachen University, 52074 Aachen, Germany}

\author{Andreas Winter}%
\affiliation{Department Mathematik/Informatik---Abteilung Informatik,\protect\\ Universit\"at zu K\"oln, Albertus-Magnus-Platz, 50923 K\"oln, Germany}
\affiliation{ICREA---Instituci\'o Catalana de Recerca i Estudis Avan\c{c}ats,\protect\\ Pg.~Lluis Companys, 23, 08010 Barcelona, Spain} 
\affiliation{ICREA {\&} Grup d'Informaci\'{o} Qu\`{a}ntica, Departament de F\'{\i}sica,\protect\\ Universitat Aut\`{o}noma de Barcelona, 08193 Bellaterra (Barcelona), Spain}
\affiliation{Institute for Advanced Study, Technische Universit\"at M\"unchen,\protect\\ Lichtenbergstra{\ss}e 2a, 85748 Garching, Germany}

\author{Ren\'{e} Schwonnek}%
\affiliation{Institut f\"{u}r Theoretische Physik, Leibniz Universit\"{a}t Hannover, 30167 Hannover, Germany}

\title{Optimal strategies for shadow tomography with limited resources}

\begin{abstract}

Shadow tomography addresses the task of efficiently predicting many expectation values of an unknown quantum state from randomized measurements on comparatively few copies. Existing analyses promise large scaling advantages, but the optimal strategies realizing these guarantees are not always known, and the required measurements are potentially challenging  to implement on current hardware. We address this gap for Pauli observables by computing optimal sample-complexity parameters and constructing optimal measurement strategies under realistic resource constraints. We focus on memoryless protocols, where each copy is measured only once, and on measurements with bounded interaction range. Our approach reduces the problem to the analysis of graph parameters of the frustration graph encoding the Pauli anticommutation relations.
We provide efficient numerical methods for the general case and
analytically prove that Clifford measurements are optimal in many situations. This includes all perfect graphs,  all single-qubit, all two-qubit measurement scenarios, and more. Applied to Hamiltonian energy estimation, our framework yields constructive strategies and improved variance bounds for molecular benchmarks.

\end{abstract}
\date{\today} 

\maketitle

{\it Introduction.---}
\newcommand\change[2]{{\color{gray}#1}#2}%
{The incompatibility} of observable quantities {is} a fundamental characteristic of quantum physics~\cite{Guhne2023Feb}. {It underlies} Bell nonlocality~\cite{Yadavalli2024Dec}, contextuality as a resource for quantum computation~\cite{Howard2014Jun}, and {certifies} randomness in quantum cryptography~\cite{Mannalatha2023Dec,Zhang2025Dec}. At the same time, incompatibility obstructs the extraction of classical information from a quantum system~\cite{Guhne2023Feb}: quantum uncertainty prohibits noiseless joint measurements of incompatible observables~\cite{Xiao2019Sep,schwonnek2018Uncertainty,Ballentine1969ky,busch2016quantum,mao2023testing} and thereby creates an overhead in sample complexity~\cite{Chen2022Jun,Anshu2021Aug}. This overhead is particularly relevant for the resource-efficient use of quantum computers, where each preparation of a circuit-output state can be costly~\cite{Huang2020Oct,Chen2024Optimal,King2025Feb,chen2022exponential}.

Shadow tomography~\cite{Aaronson2020Jun,Huang2020Oct} provides a powerful framework for mitigating this overhead. Its basic idea is to predict many target observables from the measurement data without reconstructing the full quantum state or measuring all observables individually. In recent years, it has become an active research direction, with much attention devoted to scalability and sample complexity~\cite{Nguyen2022Nov,Chen2024Optimal,King2025Feb}. However, concrete optimal strategies adapted to realistic measurement constraints remain limited~\cite{Rozon2024Sep,Chen2024Optimal}. In this work, we address this gap for Pauli observables and show that, in many relevant cases, optimality can already be achieved with efficiently implementable Clifford measurements.

For a general set of observables, a systematic characterization of optimal measurement strategies appears out of reach. We therefore focus on Pauli observables on $n$ qubits, which are central to quantum computation, quantum error correction, and the description of many-body Hamiltonians in quantum simulation~\cite{Gottesman1997Mar,Huang2020Oct,Loizeau2025Mar}. The concrete task for this work is as follows:
For Pauli observables $\mathcal{S}=\{S_1,\ldots,S_v\}$ and a source or circuit that produces an unknown $n$-qubit state $\rho$, find a measurement strategy that allows us to estimate all expectation values $\operatorname{tr}(S_i\rho)$ up to additive error $\epsilon$.
The efficiency of a strategy is quantified by the number $N$ of independently prepared copies of $\rho$ required to reach this accuracy. We consider this task in a memoryless setting~\cite{Huang2021May,Huang2020Oct,chen2022exponential}. That is, each copy of $\rho$ can only be measured once, and no coherent quantum memory across different copies is available. This setting is highly relevant, as it reflects the realistic limitations of existing quantum hardware~\cite{Aaronson2020Jun,Huang2020Oct,chen2022exponential}. An important recent result showed that the sample complexity of optimal memoryless strategies scales as $N \sim \tilde{\Theta}\!\left(\frac{1}{\epsilon^2\,\delta(\{S_i\})}\right)$, governed by a parameter $\delta$~\cite{Chen2024Optimal}. However, the evaluation of $\delta$ remains an open issue, let alone the development of an optimal strategy.

Extending this direction, we further investigate measurement strategies with a restricted interaction range, where entangling measurements can only be performed on blocks of $m$ out of $n$ qubits. This captures typical limitations of near-term quantum-computers~\cite{Preskill2018Aug}, where entangling operations are constrained by connectivity and gate depth. We show that, in such a block-limited setting, described by a partition $\tau$, a similar scaling law holds with a parameter $\delta_\tau$.

In this work, we provide strategies that attain this scaling. A crucial step is to develop reliable and efficient methods for computing these parameters on concrete instances. Our central object is the graph $G$ that encodes the anticommutativity relations in a given set $\mathcal{S}$ of Pauli strings.
The basic intuition is that this graph precisely captures the obstruction to extracting several Pauli expectation values from the same copy. The complexity parameter $\delta$ is closely linked to a graph invariant, $\beta(G)$, introduced in our prior work~\cite{Xu2024Apr,xu2025simultaneous}. Despite finding optimal strategies for the general case via see-saw methods and hierarchies of semidefinite programs, instances for which this graph parameter coincides with the independence number of $G$ play a special role. In a companion work, we call them $\hbar$-perfect graphs~\cite{xu2025simultaneous}.
We find that, in this case, the optimal strategy can be realized by Clifford measurements, i.e., by Clifford unitaries followed by measurements in the computational basis.
We furthermore show that Clifford measurements are also optimal when at most two qubits can be measured jointly, with single-qubit measurements as special cases.

These results are important from the perspective of realistic resources. On NISQ devices~\cite{Preskill2018Aug,Bharti2022Feb}, Clifford measurements can usually be implemented as native operations on a fundamental hardware level. In stabilizer error correction, they coincide with the natural measurement primitives and can be implemented without magic.
As a demonstration, we construct measurement strategies for the Hamiltonians of several simulated molecules. A comparison with the best reported method in the literature~\cite{McNulty2023Mar} shows that our results can improve performance by between 20\% and 200\%. 

\subsection*{Computing the Sample Complexity}

{\it Frustration graphs, beta numbers, and $\hbar$-perfectness.---}
A Pauli observable~\cite{Nielsen2010Dec} on $n$ qubits is given by a tensor product $S=\bigotimes_{k=1}^{n} \sigma_{\mu_k}$ of Pauli matrices, i.e. $\sigma_{\mu_k}\in\{\id,\sigma_x,\sigma_y,\sigma_z\}$. 
Any two Pauli observables either commute or anticommute. For a given set of Pauli observables $\{S_i\}$, its frustration graph $G$ has a vertex $i$ for every observable $S_i$ and an edge $(i,j)$ if and only if $S_i$ anticommutes with $S_j$. 
Conversely, for every graph $G$, we can find sets of Pauli observables with $G$ as the frustration graph~\cite{Chapman2020Jun}. We call these sets realizations of $G$, which are generally not unique.   
In  \cite{Xu2024Apr} we show that the quantity 
\begin{align}\label{eq:def-beta}
 \beta(G,w)=\max_{\rho} \sum\nolimits_i w_i \langle S_i\rangle_{\rho}^2,
\end{align}
with weights $w_i\geq0$, takes the same value for any realization of $G$. This makes it a graph parameter, we refer to it as the weighted beta number. The parameter  $\beta(G,w)$ is lower-bounded  by the weighted independence number $\alpha(G,w)$ and upper-bounded by the weighted Lov\'asz number $\vartheta(G,w)$ \cite{Xu2024Apr,deGois2023Jun,Hastings2022Jun}. 

A graph $G$ is said to be $\hbar$-perfect if $\beta(G,w) = \alpha(G,w)$ for any non-negative weight vector $w$.
The stabilizer set polytope $\stab(G)$ is the set of points $u$ such that $\sum_i w_i u_i \le \alpha(G,w)$ for every non-negative weight vector $w$. Similarly, $\be(G)$ is the set with upper bound $\beta(G,w)$ and $\th(G)$ is the set with upper bound $\vartheta(G,w)$. Hence, $\stab(G)\subseteq \be(G) \subseteq \th(G)$, and a graph $G$ is $\hbar$-perfect precisely when $\be(G)=\stab(G)$.

{\it Sample complexity parameter.---}
The very first step in an optimal shadow tomography strategy is to determine the sample complexity parameter~\cite{Chen2024Optimal}, which can be expressed in terms of the weighted beta number as 
\begin{equation}\label{eq:reform}
  \delta(\{S_i\}) = \min_{w\in{\cal D}} \beta(G,w),
\end{equation}
where $\mathcal{D}$ is the set of probability vectors.
Hence, the sample complexity parameter $\delta(\{S_i\})$ is in fact identical across all realizations of a frustration graph and by this a graph parameter itself. 
This has immediate consequences for computing $\delta(\{S_i\})$. Instead of $\{S_i\}$, we can consider a new set of Pauli observables $\{S_i^{\min}\}$ sharing the same frustration graph but with a minimal amount of qubits. This set can be actually be constructed efficiently~\cite{Samoilenko,Xu2024Apr}. 
As we shall see, we can apply the same reduction in the construction of shadow tomography strategies, which can lower the computational cost.

The reformulation in Eq.~\eqref{eq:reform} does not yet provide an efficient method to compute $\delta$ directly, since evaluating $\beta(G,w)$ generally requires a semidefinite programming (SDP) hierarchy~\cite{Moran2024May,xu2025simultaneous} followed by the minimization over the weight vector $w$. Nevertheless, %
there is an SDP hierarchy for $\delta(\{S_i\})$, which turns out to be efficient for small frustration graphs:  
Let $\lambda_r(G,w)$ denote the optimal value of the $r$-th level in the complete hierarchy converging to $\beta(G,w)$. We prove that the feasible set of non-negative weight vectors $w$ satisfying $\lambda_r(G,w)\le 1$ can be characterized by an SDP. Let $\omega_r$ denote the maximal value of $\sum_i w_i$ over feasible $w$. Then $1/\omega_r$ approximates the sample complexity parameter $\delta(\{S_i\})$ as $r$ goes to infinity, since $\beta(G,w)$ is a homogeneous function of degree $1$ in $w$.
More details on the derivation and performance of this SDP hierarchy are provided in Supplemental Material (SM) Sec.~\ref{sec:A}. 

When the frustration graph has more than $10$ vertices, this hierarchy becomes heavy after the second level. Nevertheless, the formulation in Eq.~\eqref{eq:reform} enables us to bound $\delta(\{S_i\})$ by two graph parameters computed by either a linear program or an SDP.
\begin{theorem}\label{thm:global}
  For a given set of Pauli observables $\{S_i\}$ with frustration graph $G$, it holds that
  \begin{equation}
    \delta(\{S_i\}) \in [1/\chi_f(G), 1/\vartheta(\bar{G})],
  \end{equation}
  where $\chi_f(G)$ is the fractional chromatic number of $G$, and $\vartheta(\bar{G})$ is the Lov\'asz number of the complement graph $\bar{G}$. Whenever $G$ is $\hbar$-perfect, the lower bound is exact and Clifford measurements are optimal.
\end{theorem}
The detailed proof is given in SM Sec.~\ref{sec:B} and follows from the formulation of $\delta(\{S_i\})$ in Eq.~\eqref{eq:reform} after replacing $\beta(G,w)$ with its lower bound $\alpha(G,w)$ and upper bound $\vartheta(G,w)$.
Since for any $\hbar$-perfect graph $G$, $\delta(\{S_i\}) = 1/\chi_f(G)$, and this value is achieved by Clifford measurements~\cite{Chen2024Optimal,King2025Feb}, it follows directly that Clifford measurements suffice for optimal shadow tomography. 

For $\hbar$-imperfect graphs, we can employ $\hbar$-perfect subgraphs to tighten $\lambda_r(G,w)$, thereby improving the performance of the hierarchy.
Symmetry in $G$ offers another way to reduce the computational cost. Since $\beta(G,w)$ is convex in $w$, we have $\beta(G,\bar{w}) \le \beta(G,w)$ where $\bar{w}$ is the symmetrization of $w$ over the automorphism group of $G$. Hence, it suffices to consider symmetric weight vectors in the hierarchy, which reduces the problem size. In particular, $\delta(\{S_i\}) = {\beta(G)}/{n}$ when $G$ is vertex-transitive, i.e., any two vertices are related by an automorphism. For example, the anticycle $\bar{C}_7$ on $7$ vertices is $\hbar$-imperfect with $\beta(\bar{C}_7) = (9+4\sqrt{2})/7$, giving $\delta(\{S_i\}) \approx 0.29912$ for any $\{S_i\}$ realizing $\bar{C}_7$.

\subsection*{Block local measurements}
As the system size grows, global entangled measurements even on a single copy of the target state become challenging in most quantum computing architectures~\cite{Baumer2021Jul,Bianchi2026Apr,Welte2021Jul}. For photonic systems~\cite{Zhang2025Oct,AghaeeRad2025Feb} this requires a lot of magic~\cite{Raussendorf2001May,Bravyi2005Feb}.  Ion based computers need many ion movements~\cite{Kielpinski2002Jun,Pino2021Apr}. Superconducting computers~\cite{you2003quantum}, which have a fixed interaction layout, can only realize global entanglement by long gate sequences. 

We therefore now 
consider the scenario where at least $k$ local measurements, each on at most $m$ qubits, can cover the $n$ qubits in the entire system.
For a given partition $\tau=({\rm part}_1,\ldots, {\rm part}_k)$, the considered measurements take the form
$\{M_1^{(o_1)}\otimes\cdots\otimes M_k^{(o_k)}\}_{o_1,\ldots,o_k}$, where $M_i$ acts on the qubits in ${\rm part}_i$ and $M_i^{(o_i)}$ is the POVM element for the outcome $o_i$.
Let $\delta_{\tau}(\{S_i\})$ denote the optimal sample complexity for global observables $\{S_i\}$ under partition $\tau$. %
\begin{lemma}
  For a given partition $\tau$, it holds that%
  \begin{equation}\label{eq:local}
    \delta_\tau(\{S_i\}) = \min_{w\in {\cal D}} \max_{\rho \in {\rm PD}_\tau} \sum\nolimits_i w_i \langle S_i\rangle_{\rho}^2,
  \end{equation}
where $\mathrm{PD}_\tau$ is the set of product states with respect to $\tau$.
\end{lemma}
The proof is given in SM Sec.~\ref{sec:C}. The same result holds without fixing a partition. Let $\delta_{k,m}$ denote the corresponding sample complexity parameter. Then the set of states to optimize over is $\cup_{\tau \in {\cal T}_{k,m}} {\rm PD}_\tau$, where ${\cal T}_{k,m}$ is the set of all partitions into $k$ parts, each of size at most $m$.

Consider the case $k=2$ with a fixed bipartition $\tau$, and assume that the observable factorizes as $S_i = T_i\otimes Q_i$ under this bipartition. The local Pauli observables for each part generally exhibit commutation and anticommutation relations that are different from those of the global Pauli observables.
Let $G_1$ and $G_2$ denote the frustration graphs of $\{T_i\}$ and $\{Q_i\}$.
Notice that the frustration graph $G$ of $\{S_i\}$ is the XOR edge union of $G_1$ and $G_2$, i.e., $(i,j)$ is an edge of $G$ if and only if it is an edge of either $G_1$ or $G_2$ but not both.
Then the sample complexity parameter $\delta_{\tau}(\{S_i\})$ corresponding to the partition $\tau$ is determined by $G_1$ and $G_2$ only. More explicitly,
\begin{align}\label{eq:bipartition}
  \delta_{\tau}(\{S_i\}) &= \min_w \max_{\sigma_1,\sigma_2} \sum\nolimits_i w_i \langle T_i\rangle_{\sigma_1}^2 \langle Q_i\rangle_{\sigma_2}^2\nonumber\\
           &= \min_w \max_{t\in \be(G_1)} \max_{q\in \be(G_2)} \sum\nolimits_i w_i t_i q_i,
\end{align}
where the first equality is from Eq.~\eqref{eq:local} and the second equality is from the definition of $\be$. 

In the case where $G_1$ and $G_2$ are $\hbar$-perfect, $\be(G_i) = \stab(G_i)$ for $i=1,2$, which are polytopes \cite{Xu2024Apr}.
For fixed $w$ and $q$, the optimal value of $\max_{t\in \stab(G_1)} w_i t_i q_i$ can always be obtained by a vertex of $\stab(G_1)$, which is the index vector of a maximal independent set of $G_1$.
A similar argument holds also for $q$. Hence, we only need to consider those $t$ and $q$ which are both index vectors corresponding to the independent sets $I_1$ in $G_1$ and $I_2$ in $G_2$, respectively.
Then the entrywise product $t\odot q$ is the index vector of $I_1\cap I_2$.

We claim that the set $\{I_1\cap I_2| I_1\in {\cal I}(G_1), I_2\in {\cal I}(G_2)\}$, after the removal of subsets of each element in it, is the set of maximal independent sets 
of $G' := G_1 \cup G_2$, where $\cup$ means the union of sets of edges. 
Firstly, for any two nodes in $I_1\cap I_2$, they are disconnected both in $G_1$ and $G_2$, hence they are still nonadjacent in $G'$.
Secondly, for any independent set of $G'$, it is automatically an independent set shared by $G_1$ and $G_2$ since they contain subsets of edges in $G'$.
This claim implies that 
  $\delta_{\tau}(\{S_i\}) = \min_{w} \max_{v \in \stab(G')} \sum_i w_i v_i$ which is $1/\chi_f(G')$ consequently.
  We remark that $\delta_{\tau}(\{S_i\})$  holds still if $G_1$ is $\hbar$-perfect, and the induced subgraph $G_2(I)$ is also $\hbar$-perfect for any maximal independent set $I\in {\cal I}(G_1)$. In the case that both $G_1$ and $G_2$ are $\hbar$-imperfect, $1/\chi_f(G')$ might only provide a lower bound for $\delta_{\tau}(\{S_i\})$.

  In comparison, $1/\vartheta(\bar{G'})$ always works as an upper bound of $\delta_{\tau}(\{S_i\})$, which is obtained by relaxing $\be(G_i)$ to $\th(G_i)$ in Eq.~\eqref{eq:bipartition}. By definition, a vector $u\in \th(G)$ if there is a state $|s\rangle$ and measurement directions $\{|v_i\rangle$\} to be an orthogonal representation of $G$, i.e., $\langle v_i|v_j\rangle = 0$ for $(i,j)$ to be an edge of $G$, such that $u_i = |\langle v_i|s\rangle|^2$. 
  For any $t\in \th(G_1)$ and $q \in \th(G_2)$, denote by $\{|u_{1,i}\rangle\}$ and $\{|u_{2,i}\rangle\}$ the corresponding orthogonal representations and by $|s_1\rangle$ and $|s_2\rangle$ the corresponding states. Then $\{|u_{1,i}\rangle\otimes |u_{2,i}\rangle\}$ is automatically an orthogonal representation of $G'$, which implies that $t\odot q \in \th(G')$ since $|(\langle u_{1,i}|\otimes\langle u_{2,i}|)(|s_1\rangle\otimes|s_2\rangle)|^2 = t_iq_i$. %
  Consequently, $\delta_{\tau}(\{S_i\}) \le \min_w \max_{s \in \th(G')} \sum_i w_i s_i$, which can be proven to be $1/\vartheta(\bar{G'})$.

  A general result can be proven in SM Sec.~\ref{sec:C}.
  \begin{theorem}\label{thm:local}
  For a given partition $\tau$ and Pauli observables $\{S_i = \otimes_{j=1}^k S_{i,j}\}$ with $S_{i,j}$ acting on the $j$-th part, 
  \begin{align}
    \delta_{\tau}(\{S_i\}) %
                           &\in [1/\chi_f(G^{(\tau)}), 1/\vartheta(\bar{G}^{(\tau)})],
  \end{align}
  where $G^{(\tau)} = \cup_{j=1}^k G_j$ and $G_j$ is the frustration graph of $\{S_{i,j}\}_i$. The lower bound is tight if either all $G_j$'s are $\hbar$-perfect, or $G^{(\tau)}$ is perfect.
\end{theorem}

\begin{corollary}\label{cor:capgs}
  For any partition $\tau$ into pairs of qubits, sample complexity parameter $\delta_{\tau} = 1/\chi_f(G^{(\tau)})$, and local Clifford measurements suffice for the optimal strategy.
\end{corollary}
This holds due to the fact that the graph of any Pauli observables on two qubits is always $\hbar$-perfect~\cite{xu2025simultaneous}.
This also implies that Pauli measurements are optimal when only measurements on a single qubit are allowed, since Clifford measurements on a single qubit are just Pauli measurements.
In the case without fixed partitions, a lower bound of $\delta_{k,m}$ is $\max_{\tau \in {\cal T}_{k,m}} 1/\chi_f(G^{(\tau)})$, which is usually not tight even if all $G^{(\tau)}$ are $\hbar$-perfect.
For a more precise estimation of $\delta_{k,m}$  in general cases, we have to employ the methods from the next section.

\subsection*{ Optimal measurement strategies}
For Pauli observables with an $\hbar$-perfect frustration graph, Theorem~\ref{thm:global} already implies that the optimal strategy with global measurements on a single copy can be realized using only Clifford measurements, as in Refs.~\cite{Chen2024Optimal,King2025Feb}. 
When the graph $G^{(\tau)}$ associated with a partition $\tau$ is $\hbar$-perfect, Clifford measurements are likewise sufficient for the optimal strategy. 
Notice that
\begin{align}\label{eq:chif_tau}
  \chi_f(G^{(\tau)}) =& \min \sum\nolimits_{I\in {\cal I}} x_I\nonumber\\
  \rm{s.t.}\ & \sum\nolimits_{I \ni i} x_I \ge 1, \,\forall i\in V,\nonumber\\
        & x_I\ge 0, \,\forall I \in {\cal I},
\end{align}
where $V$ is the vertex set, and ${\cal I}$ is the set of all maximal independent sets of $G^{(\tau)}$.
Let $(x_I \mid I\in{\cal I})$ be an optimal solution of the linear program~\eqref{eq:chif_tau}. 
Adopting the notation of Theorem~\ref{thm:local}, a maximal independent set $I$ of $G^{(\tau)}$ is automatically an independent set of every $G_j$, which corresponds to a local Clifford measurement $M_{I,j}$ on the $j$-th part for each $j$. The optimal sample complexity parameter is then achieved by implementing the measurement $\{M_{I,1}^{(o_1)}\otimes \cdots \otimes M_{I,k}^{(o_k)}\}_{o_1,\ldots,o_k}$ with probability $x_I/\chi_f(G^{(\tau)})$. Further details appear in SM Sec.~\ref{sec:D}.

In the general case, Clifford measurements are insufficient, and we develop numerical methods to construct optimal strategies.
By the minimax theorem~\cite{Chen2024Optimal}, the sample complexity parameter for partition $\tau$ can be reformulated as
\begin{align}
  \delta(\{S_i\}) &= \max_{\{(q_l,\rho_l)\}\in {\cal E}_\tau }\min_w \sum\nolimits_{i,l} w_i q_l \langle S_i\rangle^2_{\rho_l}\nonumber\\
                  &= \max_{\gamma\in {\rm SEP}_{\tau,2}} \min_w \sum\nolimits_{i} w_i \langle S_i\otimes S_i \rangle_{\gamma} 
\end{align}
where ${\cal E}_{\tau}$ is the set of ensembles of separable states with respect to partition $\tau$, and ${\rm SEP}_{\tau,2}$ is the convex hull of the states $\rho_l\otimes\rho_l$ for separable $\rho_l$ under partition $\tau$.
Once we have obtained the optimal ensemble $\{(q_l,\rho_l)\}$, the optimal measurement strategy follows from it~\cite{Chen2022Jun}.%

Relaxing ${\rm SEP}_{\tau,2}$ via symmetric extensions and PPT conditions~\cite{Horodecki1996Nov,Peres1996Aug} yields a hierarchy that approximates the optimal sample complexity parameter from above, as explained in SM Sec.~\ref{sec:A}.
After symmetry reduction, the problem size grows polynomially with the hierarchy level. 
The quantum de Finetti theorem~\cite{Stormer1969Feb,Caves2002Sep,Christandl2007Jul} then guarantees the existence of a state in ${\rm SEP}_{\tau,2}$, i.e., a strategy with sample complexity parameter at least $\delta - \sum_{i=1}^k 4(2d_i -\max_i d_i)/(m+2)$.
Although this lower bound is impractical, it ensures theoretical convergence of the corresponding strategy. 
Further refinements from the quantum de Finetti theorem may help to tighten this bound. 
To reduce the computational cost, one can work with the shortest realization of the frustration graph during the construction of the optimal strategy. The number of qubits in the $m$-th level of the hierarchy then drops from $(m+2)n$ to $(m+2){\rm rank}_{\mathbb{F}_2}(A)/2$, where $A$ is the adjacency matrix of the frustration graph and ${\rm rank}_{\mathbb{F}_2}$ denotes the rank over the field $\mathbb{F}_2$.
For all the 986 non-isomorphic graphs with no more than $7$ vertices, $m=1$ is enough for the optimal sample complexity parameter.

However, when the number of qubits is greater than $4$, which could correspond to graphs up to 256 vertices, this method becomes heavy. Alternatively, we can use the SDP hierarchy based on state polynomial optimisation \cite{klep2024state,xu2025simultaneous} for bounding $\lambda_r(G,w)$ from above. 
In comparison, the size of SDPs in this hierarchy is determined by the number of Pauli observables instead of the number of qubits, which is more friendly for long Pauli observables. However, it is not so direct to extract a state from the optimal solution of some SDP arising from this hierarchy, which can be complemented by a multiplicative weight update (MWU)  method to construct the optimal strategy. 
Moreover, this MWU method extends naturally to scenarios with further restricted measurement capabilities, by adapting the set of ${\rm PD}_{\tau}$ in Eq.~\eqref{eq:local} accordingly.
The detailed description of this MWU method and its convergence analysis are provided in SM Sec.~\ref{sec:D}.

\subsection*{Application: Estimation of many body Hamiltonians}
A typical application of shadow tomography is the energy estimation of an unknown state~\cite{Huang2020Oct,Hadfield2022May,hadfield2021adaptive,korhonen2026improving}. This task plays a central role in quantum variational algorithms~\cite{Peruzzo2014Jul,Kandala2017Sep}, quantum chemistry~\cite{OMalley2016Jul,Hempel2018Jul} and serves here as a benchmark for different measurement strategies.
Given a Hamiltonian $H$ expanded in Pauli observables and a state $\rho$, the goal is to minimize the sample complexity for estimating $\tr(H\rho)$. We construct measurement strategies for different scenarios that not only improve efficiency over existing results~\cite{Hadfield2022May,McNulty2023Mar}, but also reduce the computational cost to construct strategies and yield an analytical upper bound via graph parameters.

As discussed above, each scenario characterized by a partition $\tau$ corresponds to a frustration graph $G^{(\tau)}$ encoding the structure of the scenario.
Our approach requires only a linear program whose size equals the number of Pauli observables in the Hamiltonian, together with an SDP of the same size. The main steps are:
\begin{enumerate}
  \item For a given Hamiltonian $H=\sum_{i=1}^n c_i S_i$ and partition $\tau$, generate $\tilde{c}_i = c_i^{2/3} /\sum_j c_j^{2/3}$ and graph $G^{(\tau)}$.
  \item Solve the linear program: $\max T$ s.t.\ $\sum_{I\ni i} t_I \ge T \tilde{c}_i$, $\forall i$, $\sum_{I\in {\cal I}} t_I = 1$, $t_I \ge 0$, where ${\cal I}$ is the set of all maximal independent sets of $G^{(\tau)}$.
  \item With the optimal solution $\{t_I\}$, implement a sharp joint measurement of Pauli observables $\{S_{i,j}\}$ for each part $j$ with probability $t_I$.
\end{enumerate}
With this approach, $N$ measurement rounds yield
\begin{align}\label{eq:varbound}
  \var(\hat{H}) 
  \le &\lambda_{\max}\Big( \sum\nolimits_{i,j} \frac{c_ic_j T_{i,j}}{T_iT_j} S_iS_j  \Big)/N \nonumber\\
  \le &\chi_f(G, \tilde{c}) \Big(\sum\nolimits_i c_i^{2/3}\Big)^{3}/N,
\end{align}
where $T_i = \sum_{I\ni i} t_I$ and $T_{i,j} = \sum_{I\ni i,j} t_I$. Table~\ref{tab:atom_variance_norm} compares our approach with the optimized joint measurement strategy of Ref.~\cite{McNulty2023Mar}. For the molecules considered, our approach achieves better estimation accuracy. Further details appear in SM Sec.~\ref{sec:E}. 

\begin{table}[htbp]
\centering
\caption{Variance bounds normalized by the value from unbiased single-qubit joint measurements, for selected active-space configurations $(e,o)$ where $e$ is the number of electrons and $o$ the number of orbitals. {\it Shadow~1/2} denotes the approach developed in this work with single-/two-qubit measurements, {\it Joint} refers to the optimized joint measurement strategy of Ref.~\cite{McNulty2023Mar}.
Smaller values indicate better accuracy.}
\label{tab:atom_variance_norm}
\begin{tabular}{llllll}
\hline\hline
Molecule & Mapper & Shadow 1 & Shadow 2 & Joint  \\
\hline
\multirow{3}{*}{H$_2$ (2, 2)}
  & Parity    & 0.10287 & 0.10287 & 0.12207 \\
  & JW        & 0.34346 & 0.30529 & 0.90055 \\
  & BK        & 0.079681 & 0.079681 & 0.094551 \\
\hline
\multirow{3}{*}{LiH (2, 3)}
  & Parity    & 0.19462 & 0.19117 & 0.25555 \\
  & JW        & 0.67105 & 0.64530 & 0.77572 \\
  & BK        & 0.23962 & 0.22565 & 0.36776 \\
\hline
\multirow{3}{*}{BeH$_2$ (4, 4)}
  & Parity    & 0.17420 & 0.16600 & 0.36968 \\
  & JW        & 0.52850 & 0.48109 & 0.75456 \\
  & BK        & 0.12479 & 0.11886 & 0.26366 \\
\hline
\multirow{3}{*}{H$_2$O (4, 4)}
  & Parity    & 0.19546 & 0.16238 & 0.35223 \\
  & JW        & 0.55741 & 0.52983 & 0.66981 \\
  & BK        & 0.13673 & 0.11368 & 0.24037 \\
\hline\hline
\end{tabular}
\end{table}

{\it Conclusion and outlook.---}
We have developed a graph-theoretic framework for memoryless shadow tomography of Pauli observables across different measurement scenarios characterized by the partition $\tau$. The sample complexity for a set of Pauli observables is determined by frustration graphs $G_j$'s in the partition and lower-bounded by the fractional chromatic number of their union graph $G^{(\tau)}$. This bound is exact for $G_j$'s to be $\hbar$-perfect or $G^{(\tau)}$ to be perfect, for which Clifford measurements suffice, most notably in the two-qubit measurement scenario.
For $\hbar$-imperfect graphs, we have developed adaptable numerical methods to approximate the optimal measurement strategy. Applied to Hamiltonian energy estimation, the framework yields variance bounds that improve upon optimized joint measurement strategies for small molecules.

Several directions remain open. A practical bottleneck is that computing the fractional chromatic number requires enumerating all maximal independent sets of the frustration graph~\cite{Chapman2020Jun}, which becomes intractable for large-molecule Hamiltonians. A promising workaround is to restrict to a manageable subset of maximal independent sets, sacrificing optimality, but preserving feasibility with rigorous variance bounds. On the theoretical side, extending the framework beyond Pauli observables and considering practical constraints would enhance its applicability. A more efficient characterization of $\hbar$-imperfect frustration graphs could also reduce the computational cost of the SDP hierarchy in the general case.
The present work lays the foundation for a systematic search for optimal shadow-tomography strategies tailored to concrete device architectures. This especially acounts for superconducting processors with fixed connectivity maps.

\section*{Acknowledgments}
The authors thank Dong-Ling Deng, Felix Huber, and Sixia Yu for inspiring discussions and comments.
Z.P.X. is supported by {National Natural Science Foundation of China} (Grant No.~12305007),
Anhui Provincial Natural Science Foundation (Grant No.~2308085QA29, No.~2508085Y003), Anhui Province
Science and Technology Innovation Project (No.~202423r06050004).
J.W. is funded by National Key R\&D Program of China under grant No.~2023YFA1009401 and Natural Science Foundation of China under grant No.~12571333.
GK acknowledges support from the Excellence Cluster - Matter and Light for Quantum Computing (ML4Q). GK acknowledges funding by the European Research Council (ERC Grant Agreement No. 948139).
A.W. is supported by the Spanish MICIN (project PID2022-141283NB-I00) with the support of FEDER funds, by the Spanish MICIN with funding from European Union NextGenerationEU (PRTR-C17.I1) and the Generalitat de Catalunya, by the Spanish MTDFP through the QUANTUM ENIA project: Quantum Spain, funded by the European Union NextGenerationEU within the framework of the ``Digital Spain 2026 Agenda'', by the Alexander von Humboldt Foundation, and by the Institute for Advanced Study of the Technical University Munich.

\onecolumngrid
\setcounter{theorem}{0}
\setcounter{secnumdepth}{4}
\renewcommand\thesection{\Alph{section}}
\section*{Supplemental Material}

The Supplemental Material is organized as follows. Section~\ref{sec:A} presents two SDP hierarchies for computing the sample complexity parameter $\delta$ based on state polynomial optimisation (SPO) and the quantum de Finetti theorem, together with the reduction to the shortest realization of the frustration graph. Section~\ref{sec:B} proves the graph-theoretic bounds of the sample complexity parameter for the memoryless scenario (Theorem~1). Section~\ref{sec:C} derives the reformulation of the sample complexity parameter in the block-local measurement scenario (Lemma~2) and its graph-theoretic bounds (Theorem~3). Section~\ref{sec:D} details the construction of optimal measurement strategies, covering both the $\hbar$-perfect case (Clifford measurements) and the general case (the see-saw--multiplicative weight algorithm). Section~\ref{sec:E} applies the framework to Hamiltonian energy estimation. 
\section{SDP hierarchies for $\delta$}
\label{sec:A}
\subsection{SDP hierarchy based on state polynomial optimisation}
Notice that $\beta(G,w)$ is a homogeneous function of degree $1$ in $w$ by definition.
Then the sample complexity parameter can be reformulated as
 \begin{align}\label{eq:delta_hier}
   \delta(\{S_i\}) &= \min_{w\ge 0, \sum_i w_i =1} \beta(G,w)\nonumber\\
                   &= \min_{w\ge 0} \beta(G,w)/\sum\nolimits_i w_i\nonumber\\
                   &= \min_{w:\, w \ge 0, \beta(G,w)\le 1} 1/\sum\nolimits_i w_i\nonumber\\
                   &= 1/\Big[\max_{w:\, w \ge 0, \beta(G,w)\le 1} \sum\nolimits_i w_i\Big].
 \end{align}
The weighted beta number $\beta(G,w)$ admits a complete hierarchy of SDP relaxations based on state polynomial optimisation~\cite{Xu2024Apr}.
At the $r$-th level, one solves
\begin{align}\label{eq:lambdar}
  \lambda_r(G,w) :=& 
\sup  \sum_{i=1}^n w_i \langle x_i\rangle^2\nonumber\\
\mathrm{s.t.}\ & M_r \succeq 0,\; [M_r]_{1,1}=1,\nonumber\\
               & [M_r]_{u,v} = [M_r]_{a,b} \text{ if } \langle u^*v\rangle = \langle a^*b\rangle,
\end{align}
where $M_r$ is the $r$-th order moment matrix indexed by state monomials of degree at most $r$, built from the non-commuting variables $x_i$ satisfying $x_i^2=1$ and $x_i x_j = -x_j x_i$ (resp.\ $x_i x_j = x_j x_i$) when $i$ and $j$ are adjacent (resp.\ nonadjacent) in the graph $G$.
It is known that $\lambda_r(G,w) \to \beta(G,w)$ from the above as $r\to\infty$~\cite{klep2024state}. Nevertheless, this complete hierarchy is computationally very expensive and scales badly with the number of vertices. A simplified and more practical hierarchy was proposed in \cite{xu2025simultaneous}.

Replacing $\beta(G,w)$ by the upper bound $\lambda_r(G,w)$ in Eq.~\eqref{eq:delta_hier}, we define
\begin{equation}\label{eq:omegar_def}
\omega_r := \max_{w\ge 0,\;\lambda_r(G,w)\le 1} \sum_i w_i,
\end{equation}
so that $\delta(\{S_i\}) = \lim_{r\to\infty} 1/\omega_r$.

We now prove that $\omega_r$ itself can be computed by a single SDP.
The SDP in Eq.~\eqref{eq:lambdar} can be written in the standard primal form
\begin{align}\label{eq:primal}
\lambda_r(G,w) = &\sup \sum_{i=1}^n w_i y_i \nonumber\\
\mathrm{s.t.}\ & M_r = A_0 + \sum_{j=1}^k y_j A_j \succeq 0,
\end{align}
where $A_0$ is the matrix with $[A_0]_{1,1}=1$ and zeros elsewhere, and the matrices $A_1,\dots,A_k$ encode the equivalence relations and Hermiticity constraints on the entries of $M_r$.
Its dual reads
\begin{align}\label{eq:dual}
\tilde\lambda_r(G,w) = &\inf \tr(Z A_0) \nonumber\\
\mathrm{s.t.}\ & \tr(Z A_i) = -w_i,\quad i=1,\dots,n, \nonumber\\
                & \tr(Z A_i) = 0,\quad i=n+1,\dots,k, \nonumber\\
                & Z \succeq 0.
\end{align}
By standard Slater-type arguments~\cite{josz2016strong}, there is no duality gap: $\tilde\lambda_r(G,w) = \lambda_r(G,w)$.

Now take any $w$ feasible for the program defining $\omega_r$ in Eq.~\eqref{eq:omegar_def}, i.e.,\ $w\ge 0$ and $\lambda_r(G,w)\le 1$.
Since $\tilde\lambda_r(G,w)=\lambda_r(G,w)\le 1$, there exists $Z\succeq 0$ satisfying the constraints in Eq.~\eqref{eq:dual} together with $\tr(Z A_0)\le 1$.
Therefore, the pair $(w,Z)$ is feasible for
\begin{align}\label{eq:omega_sdp}
\tilde\omega_r := &\sup \sum_{i=1}^n w_i \nonumber\\
\mathrm{s.t.}\ & \tr(Z A_i) = -w_i,\quad i=1,\dots,n, \nonumber\\
                & \tr(Z A_i) = 0,\quad i=n+1,\dots,k, \nonumber\\
                & \tr(Z A_0) \le 1, \nonumber\\
                & w_i \ge 0,\quad Z \succeq 0.
\end{align}
Conversely, any feasible $(w,Z)$ of Eq.~\eqref{eq:omega_sdp} yields a feasible $Z$ for Eq.~\eqref{eq:dual} with $\tr(Z A_0)\le 1$. Hence $\lambda_r(G,w)=\tilde\lambda_r(G,w)\le 1$, and so $w$ is feasible for Eq.~\eqref{eq:omegar_def}.
Thus, $\omega_r = \tilde\omega_r$ which could be obtained by solving the SDP~\eqref{eq:omega_sdp}.

In summary, $1/\omega_r$ provides a converging upper bound on $\delta(\{S_i\})$:
\begin{equation}
\delta(\{S_i\}) = \frac{1}{\omega_\infty} \le \frac{1}{\omega_r},
\end{equation}
with the approximation improving monotonically as $r$ increases.
In particular, for all $986$ non-isomorphic graphs with at most $7$ vertices, the $r=1$ level of this hierarchy already yields the exact value of $\delta(\{S_i\})$. More details on the performance is discussed in Sec.~\ref{sssec:smmwu}.

\subsection{SDP hierarchy via quantum de Finetti theorem}
Quantum de Finetti theorem ensures the approximation of the set of separable states with the set of reduced states of symmetric states in a large system.
This induces our second SDP hierarchy, which not only provides upper bound for the complexity parameters $\delta$ and $\delta_{\tau}$, but also leads to effective measurement strategies. The $m$-th level of this hierarchy reads:
\begin{align}
  \mu_m :=&\max_{\gamma} \mu\nonumber\\
  \rm{s.t.}\ & \tr[(S_i\otimes S_i)\gamma] \ge \mu,\, \forall i,\nonumber\\
        & \gamma_{m+2} \in {\rm Sym}^{m+2}({\cal H}) \cap {\rm PPT},\nonumber\\
        & \gamma_{m+2} \succeq 0, \,\tr(\gamma)=1,
\end{align}
where ${\rm Sym}^{m+2}({\cal H})$ is the symmetric subspace of ${\cal H}^{\otimes (m+2)}$ and ${\cal H}$ is the Hilbert space of $\{S_i\}$, $\gamma$ is the reduced state of $\gamma_{m+2}$ on the first two ${\cal H}$'s, ${\rm PPT}$ is the set of PPT states for any bipartition.

The quantum de Finetti theorem~\cite[Thm. II.8]{Christandl2007Jul} ensures that there is a separable state $\tau = \sum_k \lambda_k \rho_k\otimes\rho_k$ with $\sum_k \rho_k=1$ such that the trace distance
$||\tau-\gamma_{12}||_1 \le 4d/(m+2)$, where $d$ is the dimension of ${\cal H}$. Then
\begin{align}
  \tr[(S_i\otimes S_i)\tau] &\ge \mu_{m} - 4d/(m+2)\\
                             &\ge \delta - 4d/(m+2),
\end{align}
where the second inequality is due to the fact that $\mu_m \ge \delta$.

Similarly, this can be generalized to the situation of local measurements.
\begin{align}
  \mu_m := &\max_{\gamma} \mu\\
  \rm{s.t.}\ & \tr[(S_i\otimes S_i)\gamma_{12}] \ge \mu, \,\forall i,\\
        & \gamma\in {\rm Sym}^{m+2}_{A_1:\cdots:A_{2n}}({\cal H}) \cap {\rm PPT},\\
       & \gamma \succeq 0,\, \tr(\gamma)=1,
\end{align}
where $A_1:\cdots:A_{2n}$ is a partition of two copies of all the qubits such that $A_{n+i}=A_i$, and ${\rm Sym}^{m+2}_{A_1:\cdots:A_{2n}}$ is the intersection of the symmetric subspaces corresponding to the $(m+2)$ copies of $A_i$ for $i=1, \ldots, 2n$, and $\gamma_{12}$ is the reduced state for the first copy of $A_1, \ldots, A_{2n}$, or equivalently, the first two copies of $A_1, \ldots, A_n$.

Then there is a separable state $\tau$ such that~\cite[Cor. 1]{Brandao2012Oct} 
\begin{equation}\label{eq:muldefinetti}
  ||\tau-\gamma_{12}||_1 \le  \sum_{i=1}^{n} \frac{8d_i}{m+2},
\end{equation}
which implies the existence of a state $\tau'=\sum_k \lambda_k \rho_k\otimes\rho_k$ with $\rho_k$ separable according to the partition $A_1:\cdots:A_n$ such that
\begin{equation}\label{eq:mullower}
  \tr[(S_i\otimes S_i)\tau'] \ge \delta - \sum_{i=1}^n 8d_i/(m+2).
\end{equation}
As $m$ goes to infinity, we obtain an ensemble for the optimal sample complexity parameter $\delta$.

More explicitly, the inequality in Eq.~\eqref{eq:muldefinetti} is due to~\cite[Cor. 1]{Brandao2012Oct}
\begin{equation}
  ||\tau-\gamma_{12}||_1 \le 4 \sum_{i=1}^{2n-1} \frac{d_i}{m+2} \le  \sum_{i=1}^{n} \frac{8d_i}{m+2}.
\end{equation}
The inequality in Eq.~\eqref{eq:mullower} follows from the following fact. Let $\tau = \sum_{k} l_k \rho_k\otimes\sigma_k$ with $\rho_k$ and $\sigma_k$ separable according to the partitions $A_1:\cdots:A_n$ and $A_{n+1}:\cdots:A_{2n}$, respectively, and set $\tau' = \sum_{k} l_k(\rho_k\otimes\rho_k + \sigma_k\otimes\sigma_k)/2$. Then the Cauchy--Schwarz inequality yields
\begin{equation}
  \tr[(S_i\otimes S_i)\tau'] \ge \tr[(S_i\otimes S_i)\tau], \,\forall i,
\end{equation}
which consequently implies the inequality in Eq.~\eqref{eq:mullower}.

\subsection{Reduction to shortest realization}
\label{sec:F}

For a given set of Pauli observables $\{S_i\}$ with the frustration graph $G$, there is a unitary $U$ such that~\cite{Xu2024Apr} 
\begin{equation}
  S_i = U (\bar{S}_i\otimes D_i) U^\dagger,
\end{equation}
where $\{\bar{S}_i\}$ is a realization of the frustration graph $G$ on ${\rm rank}_{\mathbb{F}_2}(A)/2$ qubits with $A$ to be the adjacency matrix of $G$ and ${\rm rank}_{\mathbb{F}_2}$ is the rank over the field $\mathbb{F}_2$, and $D_i$'s are diagonal matrices with diagonal elements to be $1$ or $-1$.

Then for a given state $\bar{\tau} = \sum_k \lambda_k \bar{\rho}_k\otimes\bar{\rho}_k$,
let
\begin{equation}
  \tau = \sum\nolimits_k \lambda_k \rho_k \otimes \rho_k,\ \rho_k = U(\bar{\rho}_k\otimes |0\rangle\langle 0|)U^\dagger.
\end{equation}
We have
\begin{align}
  \tr[(S_i\otimes S_i)\tau] &= \sum\nolimits_k \lambda_k \tr[S_i \rho_k]^2\\
                            &= \sum\nolimits_k \lambda_k \tr[(\bar{S}_i\otimes D_i)(\bar{\rho}_k\otimes |0\rangle\langle 0|)]^2\\
                            &= \sum\nolimits_k \lambda_k \tr[\bar{S}_i\bar{\rho}_k]^2 \langle 0|D_i|0\rangle^2\\
                            &= \tr[(\bar{S}_i\otimes\bar{S}_i)\bar{\tau}].
\end{align}
Hence, the problem reduces to finding a state $\bar{\tau} = \sum_k \lambda_k \bar{\rho}_k\otimes\bar{\rho}_k$ such that $\tr[(\bar{S}_i\otimes\bar{S}_i)\bar{\tau}] \ge \delta - \epsilon$. This simplifies the problem, especially when ${\rm rank}_{\mathbb{F}_2}(A)/2 \ll n$ with $n$ to be the number of qubits in the support of $S_i$'s.

Similarly, we could consider the situation where we only have local measurements. In this situation, we consider the reduction of $\{S_i^{(k)}\}$ one by one, where $S_i^{(k)}$ is the $k$-th part in the partition of $S_i$, i.e., $S_i = \otimes_k S_i^{(k)}$.

\section{Bounds for the sample complexity parameter}
\label{sec:B}

\begin{theorem}
  For a given set of Pauli observables $\{S_i\}$ with $G$ to be its frustration graph, it holds that
  \begin{equation}
    \delta(\{S_i\}) \in [\,1/\chi_f(G),\; 1/\vartheta(\bar{G})\,],
  \end{equation}
  where $\chi_f(G)$ is the fractional chromatic number of $G$, and $\vartheta(\bar{G})$ is the Lov\'asz number of the complement graph $\bar{G}$.
  Whenever $G$ is $\hbar$-perfect, the lower bound is exact and Clifford measurements are optimal.
\end{theorem}

Recall from the main text that the sample complexity parameter can be expressed via the beta number as
\begin{equation}\label{eq:delta_beta}
\delta(\{S_i\}) = \min_{w\in\mathcal{D}} \beta(G,w),
\end{equation}
where $\beta(G,w) = \max_{\rho} \sum_i w_i \langle S_i\rangle_\rho^2$ and $\mathcal{D}$ is the set of probability vectors.

It is known that for any non-negative weight vector $w$, the weighted beta number is bounded by the weighted independence number and the weighted Lov\'asz number~\cite{Xu2024Apr,deGois2023Jun,Hastings2022Jun}:
\begin{equation}\label{eq:beta_bounds}
\alpha(G,w) \le \beta(G,w) \le \vartheta(G,w), \quad \forall\, w\ge 0.
\end{equation}

We now relate the minima of the two bounding quantities to graph parameters.
For the lower bound,
\begin{align}
\min_{w\in\mathcal{D}} \alpha(G,w)
  &= \min_{w\ge 0} \frac{\alpha(G,w)}{\sum_i w_i} \nonumber\\
  &= \min_{w\ge 0,\;\alpha(G,w)=1} \frac{1}{\sum_i w_i} \nonumber\\
  &= \frac{1}{\alpha^*(\bar{G})} = \frac{1}{\chi_f(G)},
\end{align}
where $\alpha^*(\bar{G})$ is the fractional packing number of the complement graph~\cite{schrijver1979fractional}, which equals the fractional chromatic number $\chi_f(G)$.
The last step follows from the definition of $\alpha^*(\bar{G})$: it is the maximum of $\sum_i w_i$ subject to $w\ge 0$ and $\sum_{i\in C} w_i \le 1$ for every clique $C$ of $\bar{G}$, i.e., every independent set of $G$.

For the upper bound,
\begin{align}
\min_{w\in\mathcal{D}} \vartheta(G,w)
  &= \min_{w\ge 0} \frac{\vartheta(G,w)}{\sum_i w_i} \nonumber\\
  &= \min_{w\ge 0,\;\vartheta(G,w)=1} \frac{1}{\sum_i w_i} \nonumber\\
  &= \frac{1}{\vartheta(\bar{G})},
\end{align}
where the last equality uses the duality relation: 
${\max_{w\ge 0,\;\vartheta(G,w)=1} \sum_i w_i} = \vartheta(\bar{G})$, which follows from
$\mathrm{TH}(\bar{G}) = \{w\mid w\ge 0,\; \sum_i w_i v_i \le 1,\;\forall v\in\mathrm{TH}(G)\}$~\cite{knuth1994sandwich}.

Combining Eqs.~\eqref{eq:delta_beta} and~\eqref{eq:beta_bounds} with the two evaluations above yields
\begin{equation}
\delta(\{S_i\}) \in [\,1/\chi_f(G),\; 1/\vartheta(\bar{G})\,].
\end{equation}

When $G$ is $\hbar$-perfect, we have $\alpha(G,w) = \beta(G,w)$ for all $w\ge 0$ by definition.
The lower bound then becomes exact:
\begin{equation}
\delta(\{S_i\}) = \min_{w\in\mathcal{D}} \alpha(G,w) = 1/\chi_f(G).
\end{equation}
Since Clifford measurements achieve the fractional chromatic number bound~\cite{Chen2024Optimal,King2025Feb}, they are optimal for any $\hbar$-perfect frustration graph.
\section{Block-local measurement scenario}
\label{sec:C}
For a given partition $\tau$, the limitation on measurements leads to potential increase of the sample complexity, i.e., a smaller sample complexity parameter $\delta_\tau$. We firstly prove a reformulation of $\delta_\tau$, and then provide the graph-theoretic bounds for it.
\subsection{Reformulation of $\delta_\tau$}
\begin{lemma}\label{lm:partition}
  For a given partition $\tau$, it holds that
  \begin{equation}\label{eq:local_sm}
    \delta_\tau(\{S_i\}) = \min_{w\in {\cal D}} \max_{\rho \in {\rm PD}_\tau} \sum\nolimits_i w_i \langle S_i\rangle_{\rho}^2,
  \end{equation}
where $\mathrm{PD}_\tau$ denotes the set of product states with respect to $\tau$.
\end{lemma}

The sample complexity of estimating a set of Pauli observables $\{S_i\}$ to precision $\epsilon$ with a measurement $M=\{F_s\}$ is governed by the $\chi^2$-divergence~\cite{Chen2024Optimal}:
\begin{equation}
\chi_M^2\!\left(\frac{I + 3\epsilon S_i}{d}\,\Big\|\,\frac{I}{d}\right)
= 9\epsilon^2 \sum_s \frac{\tr(S_i F_s)^2}{d\,\tr(F_s)},
\end{equation}
where $d=2^n$ is the Hilbert-space dimension.
The optimal sample complexity parameter, up to the factor $9\epsilon^2$, is therefore
\begin{equation}\label{eq:chidef}
\delta(\{S_i\}) = \min_{w\in\mathcal{D}} \max_{M\in\mathcal{M}} \sum_i w_i \sum_s \frac{\tr(S_i F_s)^2}{d\,\tr(F_s)},
\end{equation}
where $\mathcal{M}$ is the set of admissible measurements.
For the global case without locality constraint, it is known that this reduces to
$\delta(\{S_i\}) = \min_{w\in\mathcal{D}} \max_\rho \sum_i w_i \langle S_i\rangle_\rho^2$~\cite{Chen2024Optimal,King2025Feb}.

We now prove the local analogue.
Fix a partition $\tau = (\mathrm{part}_1,\dots,\mathrm{part}_k)$ and let $\mathcal{M}_\tau$ be the set of measurements that are tensor products of POVMs on each part.
We claim that for any weight vector $w$,
\begin{equation}\label{eq:cl}
\max_{M\in\mathcal{M}_\tau} \sum_i w_i \sum_s \frac{\tr(S_i F_s)^2}{d\,\tr(F_s)}
= \max_{\rho\in\mathrm{PD}_\tau} \sum_i w_i \langle S_i\rangle_\rho^2.
\end{equation}

\textit{Proof of the inequality $\le$.}\enspace
For any measurement $M=\{F_s\}\in\mathcal{M}_\tau$, each POVM element is of the form $F_s = \bigotimes_{j=1}^k F_{s,j}$ with $F_{s,j}$ acting on $\mathrm{part}_j$.
Consequently $\rho_s := F_s/\tr(F_s)$ is a product state, i.e., $\rho_s\in\mathrm{PD}_\tau$.
Then
\begin{align}
\sum_i w_i \sum_s \frac{\tr(S_i F_s)^2}{d\,\tr(F_s)}
&= \sum_s \frac{\tr(F_s)}{d} \sum_i w_i \tr(S_i \rho_s)^2 \nonumber\\
&\le \sum_s \frac{\tr(F_s)}{d} \max_{\rho\in\mathrm{PD}_\tau} \sum_i w_i \langle S_i\rangle_\rho^2 \nonumber\\
&= \max_{\rho\in\mathrm{PD}_\tau} \sum_i w_i \langle S_i\rangle_\rho^2,
\end{align}
since $\sum_s \tr(F_s)/d = 1$ due to $\sum_s F_s = d\,\id$.
Thus the left-hand side of~(\ref{eq:cl}) never exceeds the right-hand side.

\textit{Proof of the inequality $\ge$.}\enspace
Let $\rho\in\mathrm{PD}_\tau$ attain the maximum on the right-hand side.
Construct the measurement $M = \left\{\frac{1}{d} Q\rho Q \mid Q\in{\cal S}_n\right\}$, where ${\cal S}_n$ is the set of all $4^n$ Pauli observables.
This measurement is a product POVM as each $Q = \bigotimes_j Q_j$ with $Q_j\in{\cal S}_{n_j}$. Hence $M\in\mathcal{M}_\tau$.
For this measurement,
\begin{align}
\sum_i w_i \sum_Q \frac{\tr(S_i\, Q\rho Q)^2}{d\cdot d}
&= \sum_i w_i \sum_Q \frac{\tr(Q S_i Q\,\rho)^2}{d^2} \nonumber\\
&= \sum_i w_i \sum_Q \frac{\tr(S_i \rho)^2}{d^2} \nonumber\\
&= \sum_i w_i \sum_Q \frac{1}{d^2} \langle S_i\rangle_\rho^2 \nonumber\\
&= \sum_i w_i \langle S_i\rangle_\rho^2,
\end{align}
where the second line uses $Q S_i Q = \pm S_i$ since Pauli observables either commute or anticommute with any Pauli string $Q$, and the sign disappears after squaring. %
Thus the left-hand side is at least the right-hand side.

Both inequalities together establish Eq.~\eqref{eq:cl}.
Substituting into Eq.~\eqref{eq:chidef} with $\mathcal{M}=\mathcal{M}_\tau$ yields
\begin{equation}
\delta_\tau(\{S_i\}) = \min_{w\in\mathcal{D}} \max_{\rho\in\mathrm{PD}_\tau} \sum_i w_i \langle S_i\rangle_\rho^2,
\end{equation}
which is the claimed formula.

The same result holds without fixing a partition, by taking the union over all admissible $\tau$. Let $\delta_{k,m}$ denote the sample complexity parameter for the case where arbitrary measurements on $m$ qubits are accessible and the system is partitioned into at least $k$ parts, and let ${\cal T}_{k,m}$ denote the set of all compatible partitions. Then
\begin{align}
  \delta_{k,m} &= \min_{w\in\mathcal{D}} \max_{\rho\in\cup_{\tau\in {\cal T}_{k,m}}\mathrm{PD}_\tau} \sum_i w_i \langle S_i\rangle_\rho^2\nonumber\\
               &=\min_{w\in\mathcal{D}} \max_{\rho\in\mathrm{SEP}_{k,m}} \sum_i w_i \langle S_i\rangle_\rho^2,
\end{align}
where ${\rm SEP}_{k,m}$ denotes the set of $k$-separable states with entanglement depth at most $m$. The second equality follows because ${\rm SEP}_{k,m}$ is the convex hull of $\cup_{\tau\in {\cal T}_{k,m}} {\rm PD}_\tau$ and the function $\langle S_i\rangle_{\rho}^2$ is convex in $\rho$.

\subsection{Graph-theoretic bounds for $\delta_\tau$}

\begin{theorem}\label{thm:local_sm}
  For a given partition $\tau$ and Pauli observables $\{S_i = \otimes_{j=1}^k S_{i,j}\}$ with $S_{i,j}$ acts on the $j$-th part,
  \begin{align}
    \delta_{\tau}(\{S_i\}) &= \min_{w} \max_{t_j \in \be(G_j)} \sum\nolimits_i w_i \prod\nolimits_{j=1}^k t_{i,j}\nonumber\\
                           &\in [1/\chi_f(G^{(\tau)}), 1/\vartheta(\bar{G}^{(\tau)})],
  \end{align}
  where $G^{(\tau)} = \cup_{j=1}^k G_j$ and $G_j$ is the frustration graph of $\{S_{i,j}\}_i$. The lower bound is tight if either all $G_j$'s are $\hbar$-perfect, or $G^{(\tau)}$ is perfect.
\end{theorem}

\begin{proof}[Proof of Theorem~3 (general $k$-partition)]
From Lemma~\ref{lm:partition} and the definition of the $\be$ body~\cite{Xu2024Apr}, for a fixed partition $\tau$ and Pauli observables $\{S_i = \otimes_{j=1}^k S_{i,j}\}$,
\begin{align}\label{eq:kpartition_delta}
  \delta_{\tau}(\{S_i\}) 
  &= \min_{w} \max_{\rho \in {\rm PD}_\tau} \sum\nolimits_i w_i \langle S_i\rangle_{\rho}^2 \nonumber\\
  &= \min_{w} \max_{\rho_j} \sum\nolimits_i w_i \prod\nolimits_{j=1}^k \langle S_{i,j}\rangle_{\rho_j}^2 \nonumber\\
  &= \min_{w} \max_{t_j \in \be(G_j)} \sum\nolimits_i w_i \prod\nolimits_{j=1}^k t_{i,j},
\end{align}
where $G_j$ is the frustration graph of $\{S_{i,j}\}_i$ and $\be(G_j) = \{v \ge 0| \sum_i w_i v_i \le \beta(G_j,w), \forall w\ge 0\}$. The last equality holds due to the multilinearity of function to optimize for the points in $\be(G_j)$.

We now prove the two bounds stated in the theorem.

\noindent\textit{Lower bound.}\enspace
When each $G_j$ is $\hbar$-perfect, we have $\be(G_j) = \stab(G_j)$ by definition~\cite{Xu2024Apr}, where $\stab(G_j)$ is the stable set polytope whose vertices are precisely the index vectors of (maximal) independent sets of $G_j$.
For fixed $w$, the objective $\sum_i w_i \prod_{j=1}^k t_{i,j}$ is multilinear in each $t_j$.
With $t_2,\dots,t_k$ fixed, the maximization over $t_1\in\stab(G_1)$ attains its optimum at a vertex of $\stab(G_1)$, i.e., at the index vector of an independent set $I_1\in\mathcal{I}(G_1)$.
Since the weight vector $w$ are non-negative, only the maximal independence sets contribute to the optimum.
Repeating this argument for $j=2,\dots,k$, we may restrict each $t_j$ to be the index vector of some maximal independent set $I_j\in\mathcal{I}(G_j)$.
The entrywise product $\bigodot_{j=1}^k t_j$ is then the index vector of $\bigcap_{j=1}^k I_j$.

We claim that the set
\begin{equation}
  \Bigl\{\bigcap\nolimits_{j=1}^k I_j \;\Big|\; I_j\in\mathcal{I}(G_j),\; j=1,\dots,k\Bigr\},
\end{equation}
after removing subsets, is exactly the set of maximal independent sets of $G^{(\tau)} := \bigcup_{j=1}^k G_j$, where the union means the union of edge sets.
Indeed, (i)~any two nodes in $\bigcap_{j=1}^k I_j$ are disconnected in every $G_j$, hence disconnected in $G^{(\tau)}$, so the intersection is an independent set of $G^{(\tau)}$;
(ii)~conversely, any independent set of $G^{(\tau)}$ is an independent set of each $G_j$ since $G_j\subseteq G^{(\tau)}$, and can therefore be expressed as $\bigcap_{j=1}^k I_j$ with $I_j$ the same set for all $j$.
Thus the collection of all such intersections, after discarding subsets, coincides with $\mathcal{I}(G^{(\tau)})$, the set of maximal independent sets of $G^{(\tau)}$.

Consequently,
\begin{align}
  \delta_{\tau}(\{S_i\}) 
  &= \min_{w} \max_{v \in \stab(G^{(\tau)})} \sum\nolimits_i w_i v_i \nonumber\\
  &= 1/\chi_f(G^{(\tau)}),
\end{align}
where the last equality follows from the same argument as in Sec.~\ref{sec:B}, i.e., the minimum over $w$ of the weighted independence number equals the reciprocal fractional chromatic number.
When only some of the $G_j$ are $\hbar$-perfect, the same reasoning holds provided the induced subgraph of each imperfect $G_j$ on any maximal independent set of the perfect components is also $\hbar$-perfect.

\noindent\textit{Upper bound.}\enspace
Relax each $\be(G_j)$ to its outer approximation $\th(G_j)$, the theta body~\cite{Xu2024Apr}.
By definition, a vector $t_j\in\th(G_j)$ admits an orthogonal representation $\{|v_{j,i}\rangle\}$ of $G_j$, i.e., $\langle v_{j,i}|v_{j,i'}\rangle=0$ whenever $(i,i')$ is an edge of $G_j$, and a handle state $|s_j\rangle$ such that $t_{j,i}=|\langle v_{j,i}|s_j\rangle|^2$.

For any $t_j\in\th(G_j)$ ($j=1,\dots,k$), the tensor-product vectors
\begin{equation}
  |w_i\rangle := |v_{1,i}\rangle \otimes |v_{2,i}\rangle \otimes \cdots \otimes |v_{k,i}\rangle
\end{equation}
form an orthogonal representation of $G^{(\tau)}$: if $(i,i')$ is an edge of $G^{(\tau)}$, then it is an edge of some $G_j$, so $\langle v_{j,i}|v_{j,i'}\rangle=0$, which forces $\langle w_i|w_{i'}\rangle = \prod_{j=1}^k \langle v_{j,i}|v_{j,i'}\rangle = 0$.
Taking the handle state $|s\rangle := |s_1\rangle\otimes\cdots\otimes|s_k\rangle$, we obtain
\begin{equation}
  |\langle w_i|s\rangle|^2 = \prod_{j=1}^k |\langle v_{j,i}|s_j\rangle|^2 = \prod_{j=1}^k t_{j,i},
\end{equation}
which shows that the entrywise product $\bigodot_{j=1}^k t_j$ belongs to $\th(G^{(\tau)})$.
Hence
\begin{equation}
  \delta_{\tau}(\{S_i\}) 
  \le \min_{w} \max_{s \in \th(G^{(\tau)})} \sum\nolimits_i w_i s_i 
  = 1/\vartheta(\bar{G}^{(\tau)}),
\end{equation}
where the last equality uses the duality relation: $\min_{w} \max_{s\in\th(G)} \sum_i w_i s_i = 1/\vartheta(\bar{G})$ (see Sec.~\ref{sec:B}).

\noindent\textit{Tightness.}\enspace
Combining the two bounds, we have
\begin{equation}
  \delta_{\tau}(\{S_i\}) \in [\,1/\chi_f(G^{(\tau)}),\; 1/\vartheta(\bar{G}^{(\tau)})\,].
\end{equation}
The lower bound is exact under either of the following conditions.
\begin{enumerate}
  \item[(i)] All $G_j$ are $\hbar$-perfect. Then we have  $\be(G_j) = \stab(G_j)$, and the multilinear vertex argument above gives $\delta_{\tau}(\{S_i\}) = 1/\chi_f(G^{(\tau)})$.
  \item[(ii)] $G^{(\tau)}$ is perfect. For a perfect graph, $\vartheta(\bar{G}^{(\tau)}) = \alpha(\bar{G}^{(\tau)}) = \omega(G^{(\tau)})$, and the fractional chromatic number equals the clique number, $\chi_f(G^{(\tau)}) = \omega(G^{(\tau)})$. Hence $\chi_f(G^{(\tau)}) = \vartheta(\bar{G}^{(\tau)})$. So the upper and lower bounds coincide and $\delta_{\tau}(\{S_i\}) = 1/\chi_f(G^{(\tau)})$.
\end{enumerate}
\end{proof}

\section{Construction of measurement strategies}
\label{sec:D}
\subsection{Measurement strategy for $\hbar$-perfect cases}

In the case of $\hbar$-perfect frustration graphs, we can circumvent the computational difficulty for optimal ensembles by constructing an optimal measurement strategy directly.
When global measurements are available, measurement strategies with the sample complexity parameter to be the corresponding fractional chromatic number exist already.
For the block-local scenario, let us take the case for a given bipartition as an example with the Pauli observables $\{S_i:=T_i\otimes Q_i\}$, where 
$\delta(\{S_i\}) = 1/\chi_f(G_1\cup G_2)$ with the two $\hbar$-perfect graphs $G_1$ and $G_2$ to be the frustration graphs of $\{T_i\}$ and $\{Q_i\}$, respectively.

Here we employ the following formulation of the fractional chromatic number with $G=G_1\cup G_2$:
\begin{align}\label{eq:chif}
  \chi_f(G) =& \min \sum_{I\in {\cal I}} x_I\\
  \rm{s.t.}\ & \sum_{I \ni i} x_I \ge 1, \,\forall i=1,\ldots, n,\\
        & x_I\ge 0, \,\forall I \in {\cal I},
\end{align}
where $n$ is the number of nodes in $G$, i.e. the number of Pauli observables in $\{S_i\}$, and ${\cal I}$ is the set of all maximal independent sets of $G$, where each element is simultaneously an independent set of $G_1$ and $G_2$.

Assume that $(x_I \mid I\in{\cal I})$ is an optimal solution to the linear program in Eq.~\eqref{eq:chif}. For each $I$, we introduce the projective measurements
\begin{align}
  {\rm PVM}(\{T_i\}_{i\in I}) &= \Big\{ \prod\nolimits_{i\in I} \frac{\id + c_i T_i}{2} | c_i = \pm 1, \forall i\in I\Big\},\\
  {\rm PVM}(\{Q_i\}_{i\in I}) &= \Big\{ \prod\nolimits_{i\in I} \frac{\id + c_i Q_i}{2} | c_i = \pm 1, \forall i\in I\Big\},
\end{align}
whose combination is a projective measurement on the whole system, denoted as ${\cal M}_I$.
By definition, the measurement ${\cal M}_I$ is a fine-graining of the projective measurement $\{(\id+S_i)/2, (\id-S_i)/2\}$ for all $i \in I$.

Then the protocol is that we implement the measurement ${\cal M}_I$ with probability $x_I/\chi_f(G_1\cup G_2)$ on the samples of an unknown state, and estimate the mean value of $S_i$ from the measurement statistics for $i\in I$.
Consequently, the variance of the estimation for $S_i$ is
\begin{equation}
  \sigma_i^2/\frac{N \sum_{I\ni i} x_I }{\chi_f(G_1\cup G_2)} \le \frac{\sigma_i^2 \chi_f(G_1\cup G_2)}{N},
\end{equation}
where $\sigma_i^2 \le 1$ is the ideal variance of $S_i$ for one round of measurements and $N$ is the total number of samples.
Thus, $N=\chi_f(G_1\cup G_2)/\epsilon^2$ samples are enough for the variance to be no more than $\epsilon^2$, which implies that this strategy is optimal. Such a strategy can be directly generalized to the situation of multipartite partitions as summarized in the main text.

\subsection{Multiplicative weight update algorithm for lower bounds of $\delta$ and optimal strategy}
\label{ssec:see-saw}
Here we take the case of global measurements for the illustration and  reformulate the sample complexity parameter, which is the starting point of the see-saw algorithm.
Let $\{S_1,\dots,S_m\}$ be a set of $n$-qubit Pauli strings. For an unknown quantum state $\sigma$, the goal of Pauli shadow tomography is to estimate all expectations $\tr(S_i\sigma)$ within error $\epsilon$ by using as few copies of $\sigma$ as possible. The optimal sample complexity is governed by the parameter~\cite{Chen2024Optimal}
\begin{align}
  \delta(\{S_i\}) &:= \min_{w\in{\cal D}_m}\; \max_{\sigma} \sum_{i=1}^m w_i \tr(S_i\sigma)^2\nonumber\\
  &= \min_{w\in{\cal D}_m}\; \max_{\rho,\sigma} \sum_{i=1}^m w_i \tr(S_i\rho)\tr(S_i\sigma),
\end{align}
where ${\cal D}_m$ is the set of probability distribution corresponding to a simplex over $\{S_1,\dots,S_m\}$, and the equality in the second line is the reformulation. 

\subsubsection{Oracle: See-saw optimisation for a fixed weight vector}
We firstly introduce the see-saw method to solve the inner maximization for a fixed probability vector $w$ as an oracle, which is just the weighted beta number of the corresponding frustration graph~\cite{Xu2024Apr},
\begin{equation}
\beta(G,w) = \max_{\rho,\sigma} \sum_i w_i \tr(S_i\rho)\tr(S_i\sigma).
\end{equation}
The steps are as following:
\begin{enumerate}
\item {Initialisation:} choose a random pure state $\rho$.
\item {Update $\sigma$:} let $M_1 = \sum_i w_i \tr(S_i\rho)\,S_i$, and set $\sigma$ to the pure state $|\psi\rangle\langle\psi|$ where $|\psi\rangle$ is the eigenvector of $M_1$ with largest eigenvalue.
\item {Update $\rho$:} let $M_2 = \sum_i w_i \tr(S_i\sigma)\,S_i$, and set $\rho$ to the maximal eigenstate of $M_2$.
\item Repeat steps~2--3 until the objective converges, e.g., the difference between two iterations is less than $10^{-5}$.
\end{enumerate}
Because each step maximises a linear functional over pure states, the objective is non‑decreasing and converges to a local optimum. 
In practice, we can run this algorithm multiple times and choose the best one. Another approach is to extract a state for initialization from the SDP hierarchy based on either SPO~\cite{xu2025simultaneous} or the quantum de Finetti theorem~\cite{Xu2024Apr}.
We take a shallow level of those hierarchies and map the optimal solution to a feasible quantum state.
When the lower bound from this see-saw method coincides with the upper bounds from SDP hierarchies, we obtain the global optimum.

\subsubsection{Algorithm: Multiplicative weight update}
To solve the outer minimisation over $w$, we employ the multiplicative weights update (MWU) method~\cite{Arora2012Multiplicative,Freund1999Adaptive}. 
For a given $T$,  set the learning rate $\eta = \sqrt{(8\ln m)/T}$.
The algorithm proceeds as follows:
\begin{itemize}
\item Start with uniform weights $w^{(0)}_i = 1/m$.
\item For $t = 1,\dots,T$:
  \begin{enumerate}
  \item Using the oracle to compute optimal states $(\rho_t,\sigma_t)$ for the current weight vector $w^{(t-1)}$.
  \item Define $c_i = \tr(S_i\rho_t)\tr(S_i\sigma_t)$.
  \item Update $w^{(t)}_i = w^{(t-1)}_i \exp(-\eta\,c_i)$ and renormalize it as a probability distribution.
  \end{enumerate}
\end{itemize}
 After $T$ iterations we collect the ensemble $\mathcal{E} = \{(\rho_t,\sigma_t)\}_{t=1}^T$.

Standard regret analysis for MWU yields the following bound (see, e.g.~\cite{Fernandez2025Sep} for a similar statistical context). Let $\overline{w} = \frac{1}{T}\sum_{t=1}^T w^{(t)}$ be the averaged weight vector. Then
\begin{equation}
\max_{\rho,\sigma}\sum_i \overline{w}_i \tr(S_i\rho)\tr(S_i\sigma)
\;\leq\; \delta(\{S_i\}) \;+\; \sqrt{{2\ln m}/{T}}.
\end{equation}
Thus, after $T = O(\epsilon^{-2}\ln m)$ iterations, the duality gap is at most $\epsilon$. 
Provided that the oracle for the inner maximisation is solved exactly, this bound ensures the convergence speed for the estimation of the sample complexity parameter $\delta$.
In practice, we terminate it when the see-saw iteration reaches a tolerance  which is negligible compared to the MWU error, e.g., $10^{-5}$.

To leverage the ensemble $\{(\rho_t,\sigma_t)\}_{t=1}^T$ further, we can solve a linear program  that finds the best convex combination:
\begin{align}
 & \max_{a_t,b_t} \mu\nonumber\\
\rm{s.t.}\ & M = \sum\nolimits_{t=1}^T \bigl(a_t\,\rho_t\otimes\rho_t + b_t\,\sigma_t\otimes\sigma_t\bigr),\nonumber\\
      & \tr(M(S_i\otimes S_i)) \ge \mu,\nonumber\\
      & a_t,b_t\ge 0,\;\sum\nolimits_t (a_t+b_t)=1.
\end{align} 
This results in a better weight vector than $\bar{w}$ in the sense that the measurement strategy of the corresponding ensemble has a lower sample complexity.
\subsubsection{Results on graphs with no more than 9 vertices}\label{sssec:smmwu}
Since we have proven that the Clifford measurement strategy is already optimal in the case that the frustration graph is $\hbar$-perfect, we only need to focus on the $\hbar$-imperfect ones.
All such $\hbar$-imperfect graphs with no more than $9$ vertices are already determined in a companion work~\cite{Xu2024Apr}.
Among all graphs with no more than $7$ vertices, antiheptagon $\bar{C}_7$ is the only case and the ensemble for its optimal measurement strategy does not require the MWU method. Since the antiheptagon $\bar{C}_7$ is vertex-transitive, we can just use the oracle to find an optimal state for $\beta(\bar{C}_7)$. Applying the unitaries that correspond to each permutation in the automorphism group of $\bar{C}_7$, we can obtain the desired ensemble for the sample complexity parameter.

For all graphs with $8$ and $9$ vertices, we firstly employ the SDP hierarchy based on SPO to obtain an upper bound of the sample complexity parameter $\delta$, and check such an upper bound against the lower bound $1/\chi_f(G)$.
If they coincide with each other, then this means that the strategy with Clifford measurements is already optimal, even though the corresponding frustration graph is $\hbar$-imperfect. Then we do not need to consider such a case.
As it turns out, only $9$ graphs with $8$ vertices and $295$ ones with $9$ vertices are left.
Then we employ the MWU method to calculate a tighter lower bound together with the ensemble to construct the corresponding measurement strategy.

Across all the left ones,
the gap  between the upper bound from the SDP hierarchy and
the lower bound from MWU method has a mean of $2.28\times10^{-5}$ and a median of
$1.62\times10^{-5}$.
The maximum gap is $2.88\times10^{-4}$, attained on the graph 
\texttt{HCQf`\~{}O\{} in graph6 format with $9$ vertices, where the upper bound is $0.2991193922$ and the lower bound is $0.2988311453$.
For all graphs with $8$ vertices the gap is below $5\times10^{-5}$, and for
$96.6\%$ of the $295$ graphs with $9$ vertices the gap is below $10^{-4}$.
The tightest lower bounds essentially match the upper bounds, e.g., the graphs
\texttt{HUzrv\~{}\}} and \texttt{HCrUqz\~{}} have gaps below $10^{-7}$.
These results demonstrate that the bounds obtained by the MWU method are close to being tight for most tested graphs, and our method explicitly constructs near‑optimal measurement ensembles that achieve the theoretical sample complexity. The code for reproducing the results is available at \url{https://github.com/wangjie212/BetaNumber}.

\section{Estimating Hamiltonian energy}
\label{sec:E}

A typical application of shadow tomography is to estimate the energy of an unknown state~\cite{Huang2020Oct,Hadfield2022May,hadfield2021adaptive,korhonen2026improving}. Given its central importance to quantum variational algorithms~\cite{Peruzzo2014Jul,Kandala2017Sep} and quantum chemistry~\cite{OMalley2016Jul,Hempel2018Jul}, we utilize this task as a benchmark for the different measurement strategies examined here.
For a given Hamiltonian $H = \sum_i c_i S_i$ expressed in terms of Pauli strings and
a state $\rho$, the aim is to minimize the sample complexity in estimating the energy
$\tr(H\rho)$.
As discussed in the main text, each measurement scenario characterized by a partition $\tau$ corresponds to a frustration graph $G=G^{(\tau)}$ encoding the structure of the scenario.
Our approach requires only a linear program whose size equals the number of Pauli
strings in the Hamiltonian, and the main steps are:
\begin{enumerate}
  \item For $H=\sum_{i=1}^n c_i S_i$ and a given partition $\tau$, generate the graph $G$
        and compute $\tilde{c}_i = c_i^{2/3}/\sum_j c_j^{2/3}$.
  \item Solve the linear program in Eq.~\eqref{eq:energyb1} below to obtain the
        optimal $\{t_I\}$.
  \item With these $\{t_I\}$, implement a sharp joint measurement of the commuting Pauli
        strings in each independent set $I$ with probability $t_I$.
\end{enumerate}
The detailed derivation of the variance bounds and the linear program are as follows.

Let ${\cal I}$ denote the set of all maximal independent sets of $G$.
Assume that we perform $t_I$ rounds of the sharp joint measurement of the commuting Pauli observables in set $I$.
Then the Pauli string $S_i$ has been estimated sharply $T_i$ times, where $T_i = \sum_{I\ni i} t_I$.
Denote $\bar{S}_{I,i}$ the sample mean for $S_i$ in context $I$, and $\hat{S}_i = (\sum_{I\ni i} \bar{S}_{I,i})/T_i$ the estimator of $S_i$.
The estimator of $H = \sum_i c_i S_i$ is constructed as $\hat{H} = \sum_i c_i \hat{S}_i$.
Let $T_{i,j} = \sum_{I\ni i,j} t_I \le \min\{T_i,T_j\} \le \sqrt{T_iT_j}$.
Then we have
\begin{align}\label{eq:varbound00}
  \var(\hat{H}) &= \sum_i \frac{c_i^2\sigma_i^2}{T_i} + 2\sum_{i<j} c_ic_j\cov(S_i,S_j) \frac{T_{i,j}}{T_iT_j}\nonumber\\
                &\le \sum_i \frac{c_i^2\sigma_i^2}{T_i} + 2\sum_{i<j} |c_ic_j|\sigma_i\sigma_j \frac{T_{i,j}}{T_iT_j}\nonumber\\
                &\le \sum_i \frac{c_i^2\sigma_i^2}{T_i} + 2\sum_{i<j} |c_ic_j|\sigma_i\sigma_j \frac{1}{\sqrt{T_iT_j}}\nonumber\\
                &= \Big(\sum_i \frac{|c_i|\sigma_i}{\sqrt{T_i}}\Big)^2
                \le \Big(\sum_i \frac{|c_i|}{\sqrt{T_i}}\Big)^2.
\end{align}
However, a direct analysis shows that
\begin{equation}\label{eq:varbound01}
  \sum_i \frac{|c_i|}{\sqrt{T_i}} \le \frac{1}{\sqrt{T}} \Big(\sum_i c_i^{2/3}\Big)^{3/2},
\end{equation}
which is attained at $T_i = T\tilde{c}_i$, where $T = \sum_i T_i, \tilde{c}_i = c_i^{2/3}/\sum_j c_j^{2/3}$.
We then maximize $T$ for a fixed total number of rounds $t = \sum_I t_I$. Without loss of generality, we normalize to $t=1$. The optimization then reads
\begin{align}\label{eq:energyb1}
  \max \,\,&T\nonumber\\
  \rm{s.t.}\,\, & \sum_{I\ni i}t_I \ge T \tilde{c}_i, \,\forall i=1,\ldots,n,\nonumber\\
       &t_I \ge 0, \,\forall I \in {\cal I},\nonumber\\
       &\sum_I t_I = 1,
\end{align}
whose optimal value is $1/\chi_f(G,\tilde{c})$.
Combining with Eqs.~(\ref{eq:varbound00},\ref{eq:varbound01}) yields
\begin{equation}\label{eq:varbound_sm}
  \var(\hat{H}) \le \chi_f(G, \tilde{c}) \Big(\sum_i c_i^{2/3}\Big)^{3}.
\end{equation}
Different measurement limitations affect the graph $G$.
If ${\cal I}$ is only a subset of all maximal independent sets, the optimal value is smaller than $1/\chi_f(G,\tilde{c})$, and we still obtain an upper bound on the right-hand side of Eq.~\eqref{eq:varbound_sm}. This yields a more practical method for the general case, circumventing the difficulty of enumerating all maximal independent sets.

Given the optimal solution $\{t_I\}$ of the linear program~\eqref{eq:energyb1}, we can estimate a second, typically tighter, upper bound on $\var(\hat{H})$.
First, note that the matrix ${\cal T}$ with $(i,j)$-th entry $T_{i,j}$ is positive semidefinite. By definition, ${\cal T} = \sum_{I}{\cal T}_I$, where the $(i,j)$-th entry of ${\cal T}_I$ equals $t_I$ if $i,j\in I$ and $0$ otherwise. The positive semidefiniteness of each ${\cal T}_I$ implies positive semidefiniteness of ${\cal T}$. Let $\tilde{{\cal T}} = D {\cal T} D^\dagger$ with $D = {\rm diag}(c_1/T_1,\ldots,c_n/T_n)$; then $\tilde{{\cal T}}$ is also positive semidefinite.
Consequently, by reformulating, for any state $\rho$,
\begin{align}\label{eq:varbound2}
  \var(\hat{H})  =& \sum_{i,j} \frac{c_ic_j T_{i,j}}{T_iT_j} [\tr(S_iS_j \rho ) - \tr(S_i\rho)\tr(S_j\rho)] \nonumber\\
  \le& \sum_{i,j} \frac{c_ic_j T_{i,j}}{T_iT_j} \tr(S_iS_j \rho ) \nonumber\\
  \le &\lambda_{\max}\Big( \sum_{i,j} \frac{c_ic_j T_{i,j}}{T_iT_j} S_iS_j  \Big),
\end{align}
where the first inequality follows because the matrix with the $(i,j)$-th entry $\tr(S_i\rho)\tr(S_j\rho)$ is positive semidefinite.

The numerical results for selected molecules are collected in Table~\ref{tab:atom_variance}. We compare the bounds from Eq.~\eqref{eq:varbound2} with global, single-qubit, and two-qubit Clifford measurements, and the bounds from joint measurement strategies~\cite{McNulty2023Mar} with and without parameter optimization.
In all cases, Eq.~\eqref{eq:varbound2} consistently outperforms the optimized joint measurement strategy, and two-qubit Clifford measurements yield tighter bounds than single-qubit ones due to the additional entanglement.
The unbiased joint measurement values serve as a baseline.
Interestingly, global Clifford measurements for H$_2$ give worse bounds than local Clifford measurements. This discrepancy may arise from relaxations in the derivation, leaving room for more efficient strategies both in achievable performance and practical construction.

\begin{table}[htbp]
\centering
\caption{Variance bounds for selected active-space configurations $(e,o)$, where $e$ is the number of electrons and $o$ the number of orbitals.
``S'' , ``S1'' and ``S2'' are the bounds from Eq.~\eqref{eq:varbound2} with global, single-qubit and two-qubit measurements, respectively;
``JM'' and ``$\overline{\rm JM}$'' refer to the joint measurement approach~\cite{McNulty2023Mar} with and without parameter optimization.
Smaller values indicate better accuracy.}
\label{tab:atom_variance}
\begin{tabular}{l l l l l l l l}
\hline\hline
Molecule $(e,o)$ & Mapper & S & S1 & S2 & JM & $\overline{\rm JM}$ \\
\hline
\multirow{3}{*}{H$_2$ (2, 2)}
  & Parity    &   1.4034 &   1.3658 &   1.3658 &   1.6207 &   13.277 \\
  & JW        &   1.4034 &   1.5366 &   1.3658 &   4.0289 &   4.4739 \\
  & BK        &   1.4034 &   1.3658 &   1.3658 &   1.6207 &   17.141 \\
\hline
\multirow{3}{*}{LiH (2, 3)}
  & Parity    &   3.1125 &   3.3682 &   3.3085 &   4.4227 &   17.306 \\
  & JW        &   3.1125 &   3.5946 &   3.4567 &   4.1553 &   5.3567 \\
  & BK        &   3.1125 &   3.6573 &   3.4441 &   5.6131 &   15.263 \\
\hline
\multirow{3}{*}{LiH (2, 4)}
  & Parity    &   18.260 &   20.140 &   19.639 &   28.245 &   66.811 \\
  & JW        &   18.260 &   19.861 &   18.942 &   20.662 &   25.568 \\
  & BK        &   18.260 &   20.168 &   19.639 &   28.371 &   100.58 \\
\hline
\multirow{3}{*}{BeH$_2$ (4, 3)}
  & Parity    &   8.6232 &   9.3248 &   9.0086 &   12.893 &   39.100 \\
  & JW        &   8.6232 &   10.000 &   9.2799 &   13.081 &   14.476 \\
  & BK        &   8.6232 &   9.7532 &   9.3901 &   17.806 &   38.316 \\
\hline
\multirow{3}{*}{BeH$_2$ (4, 4)}
  & Parity    &   8.6880 &   10.011 &   9.5400 &   21.245 &   57.469 \\
  & JW        &   8.6880 &   10.488 &   9.5469 &   14.974 &   19.844 \\
  & BK        &   8.6880 &   10.016 &   9.5400 &   21.163 &   80.265 \\
\hline
\multirow{3}{*}{BeH$_2$ (4, 5)}
  & Parity    &   46.316 &   62.024 &   54.794 &   110.00 &   248.92 \\
  & JW        &   46.316 &   49.695 &   48.609 &   58.370 &   173.42 \\
  & BK        &   46.316 &   59.008 &   53.993 &   137.41 &   256.29 \\
\hline
\multirow{3}{*}{H$_2$O (2, 3)}
  & Parity    &   20.902 &   21.665 &   21.318 &   26.316 &   101.78 \\
  & JW        &   20.902 &   22.992 &   21.851 &   27.324 &   35.224 \\
  & BK        &   20.902 &   22.564 &   22.243 &   32.851 &   93.983 \\
\hline
\multirow{3}{*}{H$_2$O (4, 4)}
  & Parity    &   22.839 &   30.838 &   25.620 &   55.572 &   157.77 \\
  & JW        &   22.839 &   30.080 &   28.592 &   36.145 &   53.964 \\
  & BK        &   22.839 &   32.983 &   27.423 &   57.986 &   241.24 \\
\hline\hline
\end{tabular}
\end{table}

\newpage
\bibliography{ref}

\end{document}